\documentclass[9pt]{sig-alternate-05-2015}
\usepackage[latin1]{inputenc}
\usepackage{amsmath}
\usepackage{graphicx}
\usepackage{subfigure}
\usepackage{amsfonts}
\usepackage{amssymb}

\usepackage{mathptmx}
\DeclareMathAlphabet{\mathcal}{OMS}{cmsy}{m}{n}

\usepackage{color,cancel} 
\usepackage{comment}

\usepackage{enumitem}
\setlist[itemize,1]{leftmargin=8pt,itemindent=3pt,itemsep=2pt,topsep=2pt,parsep=1pt,partopsep=1pt}
\setlist[enumerate,1]{leftmargin=8pt,
  itemindent=3pt,
  itemsep=2pt,topsep=2pt,parsep=1pt,partopsep=1pt}
\setlist[description,1]{leftmargin=5pt,itemindent=3pt,itemsep=2pt,topsep=2pt,parsep=1pt,partopsep=1pt}

\usepackage{siunitx}

\usepackage{tikz}
\usepackage{pgfplots}

\usetikzlibrary{arrows,calc,3d,matrix,intersections,external,fadings,patterns,decorations.pathreplacing,positioning}

\colorlet{darkgreen}{green!50!black}
\colorlet{color1}{blue!70!black}
\colorlet{color2}{red!70!black}
\colorlet{color3}{green!65!black}
\colorlet{myblue}{color1}
\colorlet{myred}{color2}
\colorlet{mygreen}{color3}
\colorlet{contourplotscolor}{black}

\tikzset{plotnode/.style={fill=white,shape=circle,inner sep=0pt,font=\small}}
\tikzset{cntnode/.style={fill=white,shape=circle,inner sep=1pt,font=\small,draw}}

\pgfplotsset{classiflabels/.style={
    xlabel={$\gamma_{2}$},
    ylabel={$\Gamma_{2}$},
    ylabel near ticks,
    every axis x label/.style={at={(ticklabel* cs:1,2pt)},anchor=north west},
    every axis y label/.style={at={(ticklabel cs:0.97)},anchor=near ticklabel}
  }}
\pgfplotsset{classifwidth/.style={width=0.45\textwidth}}

\renewcommand{\thesubfigure}{\thefigure.\arabic{subfigure}}
\makeatletter
\renewcommand{\@thesubfigure}{\thesubfigure:\space}
\renewcommand{\p@subfigure}{}
\makeatother

\newtheorem{theorem}{Theorem}
\newtheorem{lemma}[theorem]{Lemma}

\newtheorem{proposition}[theorem]{Proposition}
\newtheorem{remark}[theorem]{Remark}

\newtheorem{definition}[theorem]{Definition}

\newcommand{\R}{\mathbb{R}}

\def\g{\ensuremath{\mathbf{g}}}
\def\mydet{\ensuremath{{D}}}

\usepackage{tikz}
\usetikzlibrary{matrix,calc,fit}
\usepackage{mathtools}
\mathtoolsset{showonlyrefs,showmanualtags}
\usepackage{booktabs}

\usepackage{url}
\usepackage{hyperref}

\newcommand{\NN}{\mathbb{N}}

\newcommand{\RR}{\mathbb{R}}
\newcommand{\QQ}{\mathbb{Q}}
\newcommand{\CC}{\mathbb{C}}

\newcommand{\PP}{\mathbb{P}}

\newcommand{\Jac}{\mathrm{Jac}}

\newcommand{\rk}{\mathrm{rank}}

\newcommand{\sing}{\mathrm{sing}}
\newcommand{\crit}{\mathrm{crit}}
\newcommand{\algores}{\mathsf{res}}
\newcommand{\card}{\#}

\newcommand{\boundary}[1]{\partial #1}
\newcommand{\closure}[1]{\overline{#1}}

\newcommand{\ind}{\hspace{5mm}} 

\newcommand{\alname}[1]{\textup{\textsf{#1}}}

\newcommand{\algocorkone}{\alname{RankExactly}}
\newcommand{\algocrit}{\alname{DeterminantCritVals}}
\newcommand{\algoboundary}{\alname{DeterminantBoundary}}

\newcommand{\funElim}{\textsf{Elimination}}
 
\newcommand{\incidvar}[1]{\mathcal{V}_{#1}}
\newcommand{\incidsys}[2]{\mathcal{F}_{#1,#2}}
\newcommand{\nummins}{N}

\newcommand{\pbox}[2][l]{%
  \begin{tabular}[c]{@{}#1@{}}#2\end{tabular}%
}

\newcommand{\hyp}[1]{$\mathcal{H}#1$}

\newcommand{\Id}[1]{\mathrm{Id}_{#1}}

\newenvironment{myalgo}{\enumerate[nosep] \small}{\endenumerate}
\newenvironment{myhypenum}{\enumerate[leftmargin=15pt]}{\endenumerate}

\newcommand{\myvspace}[1]{}

\usepackage[f]{esvect}

\newcommand{\pfrac}[2]{\frac{\partial #1}{\partial #2}}

\begin{document}

\CopyrightYear{2016} 
\setcopyright{acmcopyright}
\conferenceinfo{ISSAC '16,}{July 19-22, 2016, Waterloo, ON, Canada}
\isbn{978-1-4503-4380-0/16/07}\acmPrice{\$15.00}
\doi{http://dx.doi.org/10.1145/2930889.2930916}



\title
{Determinantal Sets, Singularities and Application\\ to Optimal Control in Medical Imagery}

\newcommand{\superscript}[1]{\ensuremath{^{\textrm{#1}}}}
\def\upmc{\superscript{a}}
\def\cnrs{\superscript{c}}
\def\inriapolsys{\superscript{d}}
\def\dijon{\superscript{b}}
\def\iuf{\superscript{f}}
\def\inriamactao{\superscript{e}}

\def\sharedaffiliation{\end{tabular}
\begin{tabular}{c}}
\def\twoauthors{\end{tabular}\newline\begin{tabular}{cc}}

\numberofauthors{4}

\author{
  \alignauthor Bernard Bonnard\dijon\inriamactao\\
  \affaddr{Bernard.Bonnard@u-bourgogne.fr}
  \alignauthor Jean-Charles Faug{\`e}re\inriapolsys\upmc\cnrs\\
  \affaddr{Jean-Charles. Faugere@inria.fr}
  \alignauthor Alain Jacquemard\inriapolsys\dijon\upmc\cnrs\\
  \affaddr{Alain.Jacquemard@u-bourgogne.fr}
  \and
  \alignauthor 
  Mohab Safey El Din\upmc\cnrs\inriapolsys
  \\
  \affaddr{Mohab.Safey@lip6.fr}
  \alignauthor Thibaut Verron\upmc\cnrs\inriapolsys\\
  \affaddr{Thibaut.Verron@lip6.fr}
  \sharedaffiliation
  \affaddr{{\upmc}Sorbonne Universit{\'e}s, UPMC Univ Paris 06, UMR 7606, LIP6,
    F-75005, Paris, France} \\
  \affaddr{{\dijon}Institut de Math{\'e}matiques de Bourgogne, UMR CNRS 5584, Dijon F-21078, France}\\
  \affaddr{{\cnrs}CNRS, UMR 7606, LIP6, F-75005, Paris, France} \\
  \affaddr{{\inriapolsys}Inria, Paris Center, PolSys Project} \\
  \affaddr{{\inriamactao}Inria, Sophia Antipolis, McTao Project} \\
}



\date{\today}



\maketitle

\begin{abstract}
  Control theory has recently been involved in the field of nuclear magnetic resonance imagery.
  The goal is to control the magnetic field optimally in order to improve the contrast between two
  biological matters on the pictures.\\
  Geometric optimal control leads us here to analyze meromorphic vector fields depending upon physical parameters, and having their singularities defined by a determinantal variety.
  The involved matrix has polynomial entries with respect to both the state variables and the parameters.
  Taking into account the physical constraints of the problem, one needs to classify, with respect to the parameters, the number of real singularities lying in some prescribed semi-algebraic set. \\
  We develop a dedicated algorithm for real root classification of the singularities of the rank defects of a polynomial matrix, cut with a given semi-algebraic set.
  The algorithm works under some genericity assumptions which are easy to check.
  These assumptions are not so restrictive and are satisfied in the aforementioned application.
  As more general strategies for real root classification do, our algorithm needs to compute the critical loci of some maps, intersections with the boundary of the semi-algebraic domain, etc.
  In order to compute these objects, the determinantal structure is exploited through a stratification by the rank of the polynomial matrix.
  This speeds up the computations by a factor 100.
  Furthermore, our implementation is able to solve the application in medical imagery, which was out of reach of more general algorithms for real root classification.
  For instance, computational results show that the contrast problem where one of the matters is water is partitioned into three distinct classes.


\end{abstract}


\section{Introduction}

\noindent {\bf Motivations and problem description.} Nuclear Magnetic Re\-so\-nan\-ce (NMR) is a powerful tool in medical imagery.
In order to distinguish two biological matters on a picture, it is required to optimize the contrast between the two matters.
Because of its importance in medical sciences, this contrast imaging problem has received a lot of attention.
The pioneering work of \cite{IEEE} has established geometric optimal control techniques as a major tool for designing optimal control strategies for the problem of improving the contrast.

These strategies depend on the biological matters under study. In NMR
imagery the main physical parameters involved are the longitudinal and
transversal relaxation times of each matter.
This approach is formalized in
\cite{MCRF13}. It requires us to solve the following real root
classification problem.

Consider a 
$k \times k$ matrix $M$ 
whose coefficients are polynomials in $\QQ[X_{1},\dots,X_{n},G_{1},\dots,G_{t}]$, and assume that $n=(k-r+1)^{2}$ with $r \in \{1,\dots,k-1\}$.
Let $\pi: \CC^n \times \CC^t \to \CC^{t}$ be the canonical projection.
Let $V_{r} \subset \CC^{n}\times \CC^{t}$ be the set of points at which $M$ has rank $r$, and let $V$ be the union of the singular locus of $V_{r}$ and of the critical points of $\pi$ restricted to $V_{r}$.
Generically, this variety has dimension $n+t-(k-r+1)^{2}=t$.
Also consider a semi-algebraic set $B$ in $\R^n\times \R^t$.
Assume that $B$ has non-empty
interior and that there exists a Zariski-open set $\mathcal{O} \subset \CC^t$ such that $V\cap \pi^{-1}(\g)$ is a non-empty finite set for $\g\in \mathcal{O}$.
Further assume that the projection $\pi$ restricted to $V \cap B$ is proper (\cite[Def.~2.10.1]{Demazure-bifurcations}).
We aim at describing connected open sets
$C_1, \ldots, C_\ell \subset \R^t$ such that $\bigcup_{i=1}^\ell C_i$ is
dense in $\R^t$ (for the Euclidean topology) and the cardinality of $V\cap \pi^{-1}(\g)$ inside $B\cap \pi^{-1}(\g)$ is invariant when $\g$ ranges over $C_i$
.

For our application, the size of the matrix is $4$, and the target rank is $k-1=3$.
The number of variables and parameters are $n=4$ and $t=4$ for the general case.
Since the system is homogeneous in the parameters, we may set one of the parameters to $1$, 
reducing the problem to $t=3$.
An important particular case is when one of the matters is water; then the corresponding relaxation times are $1$, and the number of parameters is $t=2$.

Such a determinantal structure is general enough to design de\-dicated algorithms.
We also mention that the optimal control problems in~\cite{ACTA15} lead to algebraic classification problems with similar structures.

\smallskip\noindent
{\bf State-of-the-art.} The modeling through an optimal control
problem is introduced in \cite{MCRF13}. The so-called Bloch modeling and
saturation method for tackling this problem is developed therein.

In \cite{MCRF13}, four experimental important cases are studied (all parameters of the classification problem are fixed).  Among other properties, it has
been observed there that the number of singularities is constant when water
is involved.
This led to the following questions:
\begin{enumerate}
  \item Is this number of singularities preserved for any choice of a second matter, the first one being water?
  \item If not, how many different classes of pairs of matters can we
  distinguish through the analysis of those singularities?
\end{enumerate}
Answering these questions leads to the real root classification problem
described above. Symbolic computation techniques are good candidates
to solve them.

Properties of Cylindrical Algebraic Decomposition\linebreak (CAD) adapted to a
given polynomial family allow to solve real root classification
problems. Hence the CAD algorithm~\cite{Collins-CAD-1975} can be used in our
context. However, the complexity of computing a CAD is doubly exponential in the number of variables
(\cite{Brown-Davenport-CAD-complexity,Davenport-Heintz-CAD-complexity}); its implementations are usually limited to non-trivial
problems involving $4$ variables and cannot tackle our application.

The complexity of computing a CAD can be much improved when ta\-king into account equational constraints (see e.g. \cite{Scott-CAD-equation}).
In the context of real root classification problems, this leads us to take advantage of the presence of equations to compute closed sets in the parameter space ($\R^t$ using our notation) containing the boundaries of the regions $C_1, \ldots, C_\ell$, hence substituting the recursive (doubly exponential) projection steps of CAD with more involved projection techniques.
In the past ten years, several works have focused on this problem~ \cite{LR07,Xia-classification-discriminant} using various computer algebra tools such as Gr\"obner bases, regular chains, etc.
We also mention \cite{RT15} which uses evaluation/interpolation techniques to compute those closed sets in the parameters space.

While the implementation of~\cite{Xia-classification-discriminant} is able to
solve our classification problem for the case of water, none of the implementations were able to classify the singular locus of $V$ in the general case (the number of parameters is $3$).

Our strategy is not as general as the ones in \cite{Xia-classification-discriminant} or \cite{LR07}. It exploits
properties of sets defined by minors of matrices with polynomial
entries. Such structures have been used for computing sample points in
each connected component of the real trace of determinantal varieties
\cite{HNSjoc16, HNShankel,HNSjsc16} or for solving linear matrix inequalities
\cite{HNSlmi}. These works are based on dedicated strategies for
computing critical loci of some projections restricted to
determinantal varieties. Such computations are naturally related to
real root classification problems and real quantifier elimination (see
e.g. \cite{HS12}). 
Finally, our
computations rely on Gr\"obner bases. Several works have shown some
connection between Gr\"obner bases and determinantal ideals~\cite{FauSafSpa13} and critical point computation~\cite{FauSafSpa12,Spa14}.

\noindent
{\bf Main results.} Our main results are twofold:
\begin{itemize}
\item an algorithm solving the special real root classification problem
  described above and which exploits the determinantal structure of
  the input data arising in contrast imaging problems;
  \item its successful use for solving the challenging application to
  the contrast problem in the general case;
  \item answers to the questions raised by the experimental data involving water: the answer to question 1. is \emph{no}, and the answer to question 2. is that there are $3$ classes of second matters that we can separate, depending on whether there are $1$, $2$ or $3$ singularities.
\end{itemize}
We start by describing our algorithmic contribution. Recall that we
are given a matrix, denoted by $M$, with polynomial entries.  As in
\cite{HNSlmi, HNSjoc16, HNShankel, HNSjsc16}, it is based on splitting
computations according to the rank of $M$.

More precisely, in order to solve our real root classification
problem, we need to identify where the number of real solutions inside $B$ of the determinantal system describing $V$ changes depending on the values of the parameters.
The real root classification problem we want to solve involves inequalities defining
a semi-algebraic set with non-empty interior. In this context, we use
standard tools from real geometry, such as Thom's first isotopy lemma,
which reduce our classification problem to computing the singular points of $V$, the critical points of the projection of the parameter space restricted to $V$, and the intersection of $V$ with the boundary of the semi-algebraic set $B$.

This computation may be difficult because generically, the variety $V_{r}$ has singularities corresponding to points where ${\rm rank}(M)<r$: this is proved using Bertini's theorem and~\cite[Prop.~1.1]{bruns2014determinantal}, as in~\cite[Prop.~2]{HNSjoc16}).
Hence, observe that the variety $V$ is naturally split according to the rank of $M$.
This is the very basic idea on which our algorithm relies: we compute critical loci of the projection on the parameter space restricted to the variety $V$ by distinguishing those points at which $M$ has rank less than $r$ from those at which $M$ has rank exactly $r$.
Incidentally, it raises the question of how these higher rank deficiencies should be interpreted from the application viewpoint.
To the best of our knowledge, this question was never raised in the optimal control community.

Our algorithms need to compute projection of algebraic sets, which is done using elimination algorithms such as Gröbner bases or triangular sets for example.
We have performed experiments for both these tools, using the package \textsf{FGb}~\cite{F10c} in Maple and an implementation of \textsf{F5}~\cite{Fau02a} for Gröbner bases, and using the package \textsf{RegularChains}~\cite{lemaire2005regularchains} in Maple for triangular sets.

Regarding the contrast imaging problem, we illustrate the behaviour of our algorithm in the case of water, giving the whole classification.
Using Gröbner bases to perform the eliminations, the computation takes \SI{10}{\second} on an \SI{2}{\giga\hertz} Intel Xeon CPU.
The \textsf{RealRootClassification} command of the Maple \textsf{RegularChain} library needs \SI{1600}{\second} to find this classification.

We also ran our algorithm on the general case.
While none of the available implementation is able to tackle this classification problem directly, ours can find the polynomials separating the open sets $C_{i}$ within \SI{4}{\hour} using \textsf{FGb}, or \SI{2}{\minute} using \textsf{F5}, and the projection step of the CAD can be done in \SI{4}{\hour}. 
We see similar speed-ups when using triangular sets to perform the elimination.

This illustrates how our dedicated algorithms take advantage of the special structure of the problem, to achieve speed-ups when compared with more general techniques.

\smallskip\noindent\textbf{Conclusion.}
We propose an algorithmic strategy refining general roots classification strategies for the case of singularities of a determinantal variety, under genericity hypotheses.
We give an overview of the results for the application at the end of the paper\footnote{The full results and the source code which produced them are available at \url{http://mercurey.gforge.inria.fr/}}.
We were able to give a full classification in the case of water, answering the questions raised by the experiments.
For the general case, the separating polynomials were found, it remains to perform the analysis of the subdivisions in order to obtain the full classification.
This work also raised questions concerning the interpretation of higher rank deficiencies from the viewpoint of control theory.


\smallskip\noindent {\bf Structure of the paper.} In
Section~\ref{sec:intromodel}, we present the mathematical background
around NMR imagery and the contrast problem.
Section~\ref{sec:Algorithm} deals with the dedicated classification
algorithm.  Finally, in Section~\ref{section:WaterSingularities}, we
report on experimental results obtained when solving the application
to the contrast imaging problem.


We acknowledge the support from the ANR-DFG research program Explosys (Grant No.ANR-14-CE35-0013-01; GL203/9-1).
Mohab Safey El Din is member of and supported by Institut Universitaire de France.

\section{Modeling the dynamics} \label{sec:intromodel}
\label{section:MathContrast}

The model we describe below has been introduced in \cite{IEEE} in order to
apply techniques from geometric optimal control theory to the control of
the spin dynamics by NMR. Up to some normalization, each spin $1/2$
particle is governed by the Bloch equation
\par\nointerlineskip
\vbox{\baselineskip=0pt\small  \begin{equation*}
    \begin{cases}
    {\dot{x}} &= -\Gamma x+u_{y}z \\
    {\dot{y}} &= -\Gamma y-u_{x}z \\
    {\dot{z}} &= \gamma(1-z)+u_{x}y-u_{y}x ,
  \end{cases}
\end{equation*}}\nointerlineskip\noindent
where the {\it state variable} $q=(x,y,z)$ represents the magnetization
vector which must lie in the \emph{Bloch ball} defined by $|q|\leq 1$, and the
{\it parameters} $(\Gamma, \gamma)$ are related to the physical relaxation
times.  The parameters must also satisfy $2\Gamma \geq \gamma >
0$.
The {\it control} $u=(u_x,u_y)$ represents the magnetic field whose
magnitude is bounded by a maximum value $\mu$. 

In the context of the contrast imaging problem, this leads to the
simultaneous control of two non-inter\-ac\-ting spins with different
relaxation time parameters. The {\it contrast by saturation method} consists in
bringing the magnetization vector of the first spin toward the center
of the Bloch ball while maximizing the modulus of the magnetization
vector of the other matter. The matter with a zero magnetization
is black on the picture, while the other matter with a maximum modulus of the
magnetization vector is bright.

Using the symmetry of revolution \cite{MCRF13} which allows to eliminate one
state variable for each matter, we obtain the system
\par\nointerlineskip
\vbox{\baselineskip=0pt\small
  \begin{equation}\label{eq:2d2}
\left\{ 
\begin{array}{rl}
 {\dot{y_1}} &= -\Gamma_1\, y_1-u_x\,z_1 \\
{\dot{z_1}} &= \gamma_1\,(1-z_1)+u_x\,y_1 \\
{\dot{y_2}} &= -\Gamma_2\, y_2-u_x\,z_2 \\
{\dot{z_2}} &= \gamma_2\,(1-z_2)+u_x\,y_2 ,
\end{array}
\right.
\quad |u| \leq \mu
\end{equation}}\nointerlineskip\noindent
and the optimal control problem is: starting from the equilibrium
point $N = ((0,1),(0,1))$, saturate the first spin, that is
$q_1(T) = 0$, where $T$ is the transfer time while maximizing
$|q_2(T)|^2$, where $|q_2(T)|$ represents the final contrast.  It is a
standard Mayer problem in optimal control,
studied in \cite{IEEE} through the
analysis of the Hamiltonian dynamics given by the Pontryagin Maximum
Principle \cite{pontryagin_book}. We 
summarize this analysis below.

Writing \eqref{eq:2d2} as
$\dot q = F(q) + u \, G(q), \quad |u| \leq \mu$, the optimality
conditions associated with the Maximum Principle lead us to construct
the optimal solution as a concatenation of {\it bang-arcs} where the
control is $u=\pm \mu$, and {\it singular arcs} solutions of
$X_e=F + u_s G$ where the control $u_s$ is the rational fraction
$-\mydet'/\mydet$ with
\begin{align*} \small 
\mydet &= \det(F, G, [G,F], [[G,F], G]) \\
\mydet' &= \det(F, G, [G,F], [[G,F], F]), 
\end{align*}
where $[\ , \ ]$ denotes the Lie bracket of vector fields.
Explicitly, with $d_{i}=\gamma_{i}-\Gamma_{i}$ ($i \in \{1,2\}$),
\begin{equation}
  \mydet = \det
  \begin{bsmallmatrix}
    -\Gamma_{{1}}y_{{1}}&-z_{{1}}-1&d_{{1}}z_{{1}}-\Gamma_{{1}}&2\,d_{{1}}y_{{1}}\\
    -\gamma_{{1}}z_{{1}}&y_{{1}}&d_{{1}}y_{{1}}&-2\,d_{{1}}z_{{1}}+\Gamma_{{1}}-d_{{1}}\\
    -\Gamma_{{2}}y_{{2}}&-z_{{2}}-1&d_{{2}}z_{{2}}-\Gamma_{{2}}&2\,d_{{2}}y_{{2}}\\
    -\gamma_{{2}}z_{{2}}&y_{{2}}&d_{{2}}y_{{2}}&-2\,d_{{2}}z_{{2}}+\Gamma_{{2}}-d_{{2}}
  \end{bsmallmatrix}.
\end{equation}

The localization of the singularities of $\{ {\mydet=0} \}$ inside the Bloch ball is important to understand the geometry of the hypersurface, as well as the dynamics of the vector field $X_e$ which is closely linked to the presence of such singularities.
Indeed, generically when approaching the surface $D = 0$ along a singular arc, the control $u = -D'/D$ is such that $|u|>\mu$.
Hence the control policy switches from singular to bang.
Therefore this surface is related to the complexity of the optimal law, as a concatenation of singular and bang arcs.
See~\cite{IEEE} for numerical simulations related to this phenomenon.



\section{Algorithm}
\label{sec:Algorithm}
\subsection{Classification strategy}
\label{sec:Class-strat}
We consider the polynomial algebra $\QQ[\mathbf{X},\mathbf{G}]$ with variables $\mathbf{X}=(X_{1},\dots,X_{n})$ and parameters $\mathbf{G}=(G_{1},\dots,G_{t})$.
Let $F$ and $H$ be families of polynomials in $\QQ[\mathbf{X},\mathbf{G}]$.
Let $V_{\RR} = V_{\RR}(F)$, $V=V_{\CC}(F)$ be respectively the set of zeroes of $F$ in $\RR^{n+t}$ and in $\CC^{n+t}$.
Let $B$ be the closed semi-algebraic set defined by $H$:
\begin{equation}
  B = \{(\mathbf{x},\mathbf{g}) \in \RR^{n+t} \mid \forall h \in H, h(\mathbf{x},\mathbf{g}) \leq 0\},
\end{equation}
and let $B_{0} = \bigcup_{h\in H} V_{\CC}(h)$.
Let $\pi : \CC^{n+t} \to \CC^{t}$ be the projection onto the affine space with coordinates $\mathbf{G}$.
Let $\sing(V)$ be the singular locus of $V$, $\crit(\pi,V)$ be the set of critical points of $\pi$ restricted to $V$, and $K(\pi,V) =  \pi(\sing(V) \cup \crit(\pi,V)) \cap \RR^{t}$.

Given a subset $A$ of a real or complex affine space, 
$\closure{A}$ and $\boundary{A}$ are used to denote 
the closure and the boundary of $A$ for the Euclidean topology respectively.

Assume that the following hypotheses are satisfied:
\begin{myhypenum}
  \item[\hyp{1}]
  There exists a nonempty Zariski-open subset $\mathcal{O}_{1}$ of $\CC^{t}$ such that for all $\mathbf{g} \in \mathcal{O}_{1}$, the fiber $V \cap \pi^{-1}(\mathbf{g})$ is a nonempty finite subset of $\CC^{n+t}$;
  \item[\hyp{2}]
  The restriction of the projection $\pi$ to $B$ is proper (\cite[Def.~2.10.1]{Demazure-bifurcations});
  \item[\hyp{3}]
  The intersection $V \cap B_{0}$ has dimension at most $t-1$ in $\CC^{n+t}$;
  \item[\hyp{4}] The variety $V$ is equidimensional of dimension $t$.
\end{myhypenum}

We want to find a non-zero polynomial $P \in \QQ[\mathbf{G}]$ such that, on each connected component $U$ of $\RR^{t} \setminus V_{\RR}(P)$, for $\mathbf{g} \in U$, the cardinality of $V \cap B \cap \pi^{-1}(\mathbf{g})$ does not depend on $\mathbf{g}$.

In Lemma~\ref{thm:classif-equiv-valcrit}, we describe a well-known strategy for computing these objects (see for example \cite{LR07,Xia-classification-discriminant}).

\begin{lemma}\label{thm:classif-equiv-valcrit}
  Let $F$ and $H$ be polynomial systems satisfying hypotheses \hyp{1}, \hyp{2}, \hyp{3} and \hyp{4}.
  Let 
  $C_{B} = \pi(V \cap B_{0})$, $U$ a non-empty connected open subset of $\RR^{t}$ which does not meet $C_{B} \cup K(\pi,V)$, and $\mathbf{g} \in U$.
  Then $V \cap \pi^{-1}(\mathbf{g})$ is finite, and $\forall\, \mathbf{g}' \in U$, $\card{\left(V \cap \pi^{-1}(\mathbf{g}')\right)} = \card{\left(V \cap \pi^{-1}(\mathbf{g})\right)}$.  
\end{lemma}
\begin{proof}
  We will construct a Whitney stratification of $V \cap B$ (\cite[Def.~9.7.1]{BCR-real-algebraic-geometry}) with certain properties.
  First note that since $V$ is $t$-equidimensional by~\hyp{4}, $V$ has real dimension at most $t$ (\cite[Prop.~2.8.2]{BCR-real-algebraic-geometry}). 
  Let $\mathcal{S}_{=t}$ be the intersection of the points where $V_{\RR}$ has local dimension $t$ 
  and of the interior of $B$.
  There exists a Whitney stratification $(\mathcal{S}_{i})$ of the semi-algebraic set $V \cap B$ such that $\mathcal{S}_{=t}$ is the union of strata of dimension $t$ (\cite[Th.~9.7.11]{BCR-real-algebraic-geometry}).
  Let $\mathcal{S}_{<t}$ be the union of the other strata, they all have real dimension less than $t$.
  By construction, this is a semi-algebraic set which is the union of $(V \cap \boundary{B}) \subset (V \cap B_{0})$ and of the singular locus of $V \cap B$, and it has dimension less than $t$ (by~\hyp{3} for $V \cap B_{0}$).
  Its image through $\pi$ has dimension less than $t$, and so it has codimension at least $1$.
    
  Now consider $\mathcal{S}_{=t}$.  By Hyp.~\hyp{1}, for any $\mathbf{g} \in \mathcal{O}_{1}$, $\pi^{-1}(g) \cap V$ is non-empty, hence $\pi(V)$
  contains the non-empty Zariski-open set
  $\mathcal{O}_{1} \subset \CC^{t}$.  The intersection
  $\mathcal{O}_{1}\cap \RR^{t}$ is a non-empty Zariski-open set of
  $\RR^{t}$, contained in $\pi(V)$, hence $\pi(V) \cap \RR^{t}$ has
  real dimension $t$.  Let $U_{0}$ be its
  interior
  .  The subset $\mathcal{S}_{=t} \cap \pi^{-1}(U_{0})$ is a locally
  closed semi-algebraic set.  If it is empty, then there is nothing to
  prove.  Otherwise, by construction it has dimension $t$; and the
  projection $\pi$ restricted to this subspace is proper, by
  Hyp.~\hyp{2}.  Thom's isotopy lemma (\cite{MR1320312}) states
  that for any nonempty connected open set $U$ of $\RR^{t}$ not
  meeting $K(\pi,V)$, and for any
  $\mathbf{g} \in U$, there exists a semi-algebraic diffeomorphism
  $h=(h_{0},\pi) : V \cap B \cap \pi^{-1}(U) \overset{\sim}{\longrightarrow}
  \pi^{-1}(\mathbf{g}) \times U$.
  By Hyp.~\hyp{1}, if $U$ is nonempty, $\pi^{-1}(\mathbf{g})$ is
  finite, and the cardinality of the fibers is constant on $U$.
\end{proof}

So computing the wanted decomposition of the parameter space is
equivalent to computing a polynomial $P \in \QQ[\mathbf{G}]$ such that
$V(P)$ covers $\pi(V \cap B_{0})$ and $K(\pi,V)$.
\subsection{The determinantal problem}
\label{sec:Stat-determ-probl}
Let $k$ be an integer greater than $1$, $r_{0} \in \{1,\dots,k-1\}$, and
$n = (k-r_{0}+1)^{2}$.  Let $t \in \NN$, and let
$M(\mathbf{X},\mathbf{G})$ be a $k \times k$ matrix with polynomial
entries in $n$ variables $\mathbf{X}=(X_{1},\dots,X_{n})$ and $t$
parameters $\mathbf{G}=(G_{1},\dots,G_{t})$.
As before, let $\pi : \CC^{n+t} \to \CC^{t}$ be the projection onto the affine space with coordinates $\mathbf{G}$.

Let $\{h(\mathbf{X},\mathbf{G}) \leq 0 \mid h\in H\}$ be a system of inequalities, with $H \subset \QQ[\mathbf{X},\mathbf{G}]$.
Let $V_{-1}=\emptyset$ by convention, and for any $r \in \{0,\dots,k\}$, we define the variety
\begin{equation}
  V_{r} = \{(\mathbf{x},\mathbf{g}) \in \CC^{n+t} \mid \rk(M(\mathbf{x},\mathbf{g})) \leq r\}
\end{equation}
and the constructible set $V_{=r} = V_{r} \setminus V_{r-1}$, that is the set of points at which the matrix $M$ has rank exactly $r$.

Let $V$ be the union of the singular locus of $V_{r_{0}}$ and of the
set of critical points of $\pi$ restricted to $V_{r_{0}}$.  We want to
classify the cardinality of the real fibers by $\pi$ of the semi-algebraic set
$V \cap \{(\mathbf{x},\mathbf{g}) \mid \forall h \in H,
h(\mathbf{x},\mathbf{g}) \leq 0\}$.

Assume that $V$ and $H$ satisfy hypotheses \hyp{1}, \hyp{2}, \hyp{3} and \hyp{4}.
Further assume that:
\begin{myhypenum}
  \item[\hyp{5}] There exists a non-empty Zariski-open subset
  $\mathcal{O}_{2}\subset
  \CC^{t}$ 
  such that
  $ V \cap \pi^{-1}(\mathcal{O}_{2}) = V_{r_{0}-1} \cap
  \pi^{-1}(\mathcal{O}_{2})$;
  \item[\hyp{6}] For any $r \in \{0,\dots,k-1\}$, the ideal defined by the $(r+1)$-minors of $M$ is radical;
  \item[\hyp{7}] For any $r \in \{0,\dots,k-1\}$, the variety $V_{r}$ is equidimensional with dimension $n+t-(k-r+1)^{2}$.
\end{myhypenum}
These properties are generic (\cite[Prop.~2]{HNSjoc16}).

\begin{lemma}
  \label{lemma:equiv-cork2-sing}
Assuming Hyp.~\hyp{6}, for any $r \in \{1,\dots,k\}$,
 $V_{r-1} \subset \sing(V_{r})$.
\end{lemma}
\begin{proof}
We will prove the following stronger statement: let $(\mathbf{x},\mathbf{g}) \in V_{r-1}$, 
then all partial derivatives of all $(r+1)$-minors of $M$ vanish at $(\mathbf{x},\mathbf{g})$.
This will prove that the Jacobian of $V_{r}$ at $(\mathbf{x},\mathbf{g})$ is the zero matrix 
, and in particular has rank $0$.
By Hyp.~\hyp{6} the ideal of all $(r+1)$-minors of $M$ is radical, so we can use the Jacobian criterion to characterize $\sing(V)$, so $(\mathbf{x},\mathbf{g}) \in \sing(V_{r})$.

Let $(\mathbf{x},\mathbf{g}) \in V_{r-1}$.
For any $(r+1)\times (r+1)$-submatrix of $M$, the result we want to prove depends only on the coefficients of the submatrix.
So \textit{w.l.o.g.}, we may assume that $r=k-1$.
  Let $(\mathbf{x},\mathbf{g}) \in V_{k-2}$, this means that all $(k-1)$-minors of $M$ vanish at $(\mathbf{x},\mathbf{g})$.
  Consider a matrix of polynomial indeterminates $\mathbf{U}$:
    \begin{equation}
    \tilde{M} = \begin{bsmallmatrix}
      U_{1,1} & \hdots & U_{1,k} \\
      \vdots & & \vdots \\
      U_{k,1} & \hdots & U_{k,k}
    \end{bsmallmatrix}
  \end{equation}
  and let $\tilde{D}$ be its determinant.  Then for any
  $i,j \in \{1,\dots,k\}$,
  $\pfrac{\tilde{D}}{U_{i,j}} = (-1)^{i+j} \cdot
  \tilde{\mathcal{M}}_{i,j}(\mathbf{U})$
  where $\tilde{\mathcal{M}}_{i,j}$ is the $(k-1)$-minor of
  $\tilde{M}$ obtained by removing row $i$ and column $j$.

  For any $i,j \in \{1,\dots,k\}$, let $m_{i,j}$ (\textit{resp.} $\mathcal{M}_{i,j}$) be the coefficient at row $i$ and column $j$ of the matrix $M$ (\textit{resp.} the minor obtained by removing row $i$ and column $j$ from $M$).
  By the derivation chain rule, for any\linebreak
  $v \in \{X_{1},\dots,X_{n},G_{1},\dots,G_{t}\}$, 
  $\pfrac{D}{v} = \sum_{i,j=1}^{k} (-1)^{i+j} \pfrac{m_{i,j}}{v} \cdot \tilde{\mathcal{M}}_{i,j}(\mathbf{m})$, which equals 
  $\sum_{i,j=1}^{k} (-1)^{i+j} \pfrac{m_{i,j}}{v} \cdot
  \mathcal{M}_{i,j}(\mathbf{X})$.
  Since by hypothesis all $(k-1)$ minors of $M$ vanish at $(\mathbf{x},\mathbf{g})$, all partial derivatives $\pfrac{D}{v}$ vanish at $(\mathbf{x},\mathbf{g})$. 
\end{proof}

In the following subsections, we will describe two algorithms \algocrit{} and \algoboundary{}, which, given such a matrix $M$, a target rank $r_{0}$ and inequalities $H$, compute a polynomial whose zeroes cover $K(\pi,V)$, and a polynomial whose zeroes cover $\pi(V \cap B_{0})$ respectively.
By Lemma~\ref{thm:classif-equiv-valcrit}, the zeroes of the product of these polynomials will subdivide the parameter space into connected components where the cardinality of real fibers is constant.
These algorithms are probabilistic, because they will rely on the choice of generic linear forms to ensure linear independence.
However, the algorithms could be made deterministic by testing that these linear forms are generic enough for our purpose, and repeating the random draw otherwise.

The algorithms will also need to compute the projection of algebraic sets onto coordinate subspaces.
For this purpose, we assume that we are given a routine \funElim, which, given a system of polynomials\linebreak $F \subset \QQ[V_{1},\dots,V_{N}]$ and a set of variables $\mathbf{V}' \subset \{V_{1},\dots,V_{N}\}$, computes a system of generators of $\langle F\rangle \cap \QQ[\mathbf{V}']$.
Such a routine can be implemented using Gröbner bases or regular chains, for example.

\subsection{Incidence varieties}
\label{sec:Incidence-varieties}
We decompose the problem depending on the rank of the matrix.
The classical technique that we use to model properties on the rank relies on  incidence varieties.
\begin{definition}
  Let $r \in \{0,\dots,k-1\}$.
  The \emph{incidence variety of rank $r$} associated with $M$ is the variety $\mathcal{V}_{r} \subset \CC^{n+t}\times (\PP^{k-1}(\CC))^{k-r}$ defined by:
  \begin{equation}
    \label{eq:incid-vectker}
    M
    \cdot
    \begin{bsmallmatrix}
      y_{1,1} & \hdots & y_{1,k-r} \\
      \vdots & & \vdots \\ 
      y_{k,1} & \hdots & y_{k,k-r}
    \end{bsmallmatrix}
    =
    \begin{bsmallmatrix}
      0 & \hdots & 0 \\
      \vdots & & \vdots \\
      0 & \hdots & 0 
    \end{bsmallmatrix}
  \end{equation}
with the additional condition that the matrix $(y_{i,j})$ has rank $k-r$.
\end{definition}

The projection of $\mathcal{V}_{r}$ onto the affine space with
coordinates $(\mathbf{X},\mathbf{G})$ is $V_{r}$.  Let
$(u_{1,1},\dots,u_{k-r,k}) \in \CC^{k(k-r)}$, we define the variety
$\mathcal{V}'_{r,\mathbf{u}}$ as the intersection of $\mathcal{V}_{r}$
and the complex solutions of
  \begin{equation}
  \label{eq:incid-vectindep}
  \begin{bsmallmatrix}
    u_{1,1} & \hdots & u_{1,k} \\
    \vdots & & \vdots \\
    u_{k-r,1} & \hdots & u_{k-r,k}
  \end{bsmallmatrix}
  \cdot
  \begin{bsmallmatrix}
    y_{1,1} & \hdots & y_{1,k-r} \\
    \vdots & & \vdots \\
    y_{k,1} & \hdots & y_{k,k-r}
  \end{bsmallmatrix}
  =
  \Id{k-r}
\end{equation}

In the rest of Section~\ref{sec:Algorithm}, $\incidsys{r}{\mathbf{u}}$
denotes the union of Equations~\eqref{eq:incid-vectker}
and~\eqref{eq:incid-vectindep}.
  

\begin{lemma}
  \label{lemma:birat-incid}
  For any $r \in \{0,\dots,k-1\}$, the varieties $V_{r}$ and $\mathcal{V}'_{r,\mathbf{u}}$ are birational (\cite[Sec.~3.4]{BCR-real-algebraic-geometry}).
\end{lemma}
\begin{proof}
  We need to define a morphism $f: W \to \mathcal{V}'_{r,\mathbf{u}}$ with $W$ a non-empty Zariski-open subset of $V_{r}$, and such that $f$ is inverse to the projection $\mathcal{V}'_{r,\mathbf{u}} \to V_{r}$ onto the affine space with coordinates $(\mathbf{X},\mathbf{G})$.
  Let $W$ be the open subset of $V_{r}$ defined as the non-vanishing locus of the top-left $r$-minor of $M$.
  Consider the block decompositions, where $A$, $Y_{(1)}$ and $U_{(1)}$ are $r \times r$ matrices:
  \par\nointerlineskip
  \vbox{\baselineskip=0pt\small  
  \begin{equation}
    M = \begin{bmatrix}
      A & B \\
      C & D
    \end{bmatrix}
    \;\;\;
    Y =
    \begin{bmatrix}
      Y_{(1)} \\ Y_{(2)}
    \end{bmatrix} \;\;\;
    U =
    \begin{bmatrix}
      U_{(1)} & U_{(2)}
    \end{bmatrix}
  \end{equation}
 .}\nointerlineskip
  Over $W$, $A$ is invertible, let $\Delta = \det(A)$.
  Let $M/A$ be the Schur complement of $A$ in $M$, Equations~\eqref{eq:incid-vectker} and~\eqref{eq:incid-vectindep} can be rewritten
  \par\nointerlineskip
  \vbox{\baselineskip=0pt\small\begin{equation}
    \label{eq:schur-Y}
    \begin{bmatrix}
      \Delta \Id{r} & A^{-1}B \\
      0 & M/A \\
      U_{(1)} & U_{(2)}
    \end{bmatrix}
    \cdot
    \begin{bmatrix}
      Y_{(1)} \\ Y_{(2)}
    \end{bmatrix}
    =
    \begin{bmatrix}
      0 \\ 0 \\ \Id{r}
    \end{bmatrix}
  \end{equation}}\nointerlineskip\noindent
  We may restrict to the open subset of $\mathcal{V}'_{r,\mathbf{u}}$ where $Y_{(2)}$ is invertible, then eliminating $Y_{(1)}$ yields that
  \par\nointerlineskip
  \vbox{\baselineskip=0pt\small\begin{equation}
    \label{eq:birat}
    \begin{cases}
      Y_{(2)} = (U_{(2)}-U_{(1)}A^{-1}B)^{-1} \\
      Y_{(1)} = \frac{-1}{\Delta} A^{-1}B Y_{(2)}
    \end{cases}
  \end{equation}}\nointerlineskip\noindent
  which defines the wanted morphism $W \to \mathcal{V}'_{r,\mathbf{u}}$.
\end{proof}


\begin{proposition}
  \label{prop:equiv-incid-sing}
  Let $r_{1}=r_{0}-1$.  Let
  $\varphi : \mathcal{V}'_{r_{1},\mathbf{u}} \to \CC^{t}$ be the
  projection onto the affine space with coordinates $\mathbf{G}$.
  Assuming that hypotheses~\hyp{1} to \hyp{7} hold, there exists a
  Zariski-open subset $\mathcal{U} \subset \CC^{k(k-r_{1})}$ such that
  if $\mathbf{u} \in \mathcal{U} \cap \QQ^{k(k-r_{1})}$,
  $\closure{K(\pi,V_{r_{1}})} =
  \closure{K(\varphi,\mathcal{V}'_{r_{1},\mathbf{u}})}$.
\end{proposition}
\begin{proof}
  Let $P=(\mathbf{x},\mathbf{g},\mathbf{y})\in \mathcal{V}'_{r_{1},\mathbf{u}}$.
  
  If $M(\mathbf{x},\mathbf{g})$ has rank less than $r_{1}$, then by Lemma~\ref{lemma:equiv-cork2-sing}, $(\mathbf{x},\mathbf{g}) \in \sing(V_{r_{1}})$, hence $\mathbf{g} \in K(\pi,V_{r_{1}})$.
  
  Since $M(\mathbf{x},\mathbf{g})$ has rank less than $r_{1}$, its kernel $L_{1}$ has dimension at least $k-r_{1}+1$.
  Equations~\eqref{eq:incid-vectindep} encode that the vectors $\mathbf{y}_{i}$ given by the columns of matrix $(y_{i,j})$ generate a $r_{1}$-dimensional vector space $L_{2}$.
  So there exists $\mathbf{y}_0 \in L_{1} \cap L_{2}$, and for all $a \in \CC$, $(\mathbf{x},\mathbf{g},\mathbf{y}_{1}+a\mathbf{y}_0,\mathbf{y}_{2})$ belongs to the fiber above $(\mathbf{x},\mathbf{g})$ in $\incidvar{r_{1}}$.
  So this fiber has dimension at least $1$, while the generic fiber has dimension $0$ by hypothesis~\hyp{1}.
  So $(\mathbf{x},\mathbf{g})$ is a critical value of the projection of $\incidvar{r_{1}}$ onto $\RR^{n+t}$, hence $(\mathbf{g}) \in K(\varphi,\mathcal{V}'_{r_{1},\mathbf{u}})$.
  
  So we may assume that $M(\mathbf{x},\mathbf{g})$ has rank exactly $r_{1}$.
  There is a $r_{1} \times r_{1}$ submatrix $A$ of $M(\mathbf{x},\mathbf{g})$ which is invertible, without loss of generality we may assume that it is the top-left $r_{1} \times r_{1}$ submatrix.
  In an open neighborhood of $(\mathbf{x},\mathbf{g})$, $V_{=r_{1}}$ is described by the vanishing of the entries of $M/A$, that is the determinants of the $(r_{1}+1)\times(r_{1}+1)$ submatrices containing $A$. 
  The same computations as in the proof of Lemma~\ref{lemma:birat-incid} give the following equations describing $\mathcal{V}'_{r_{1},\mathbf{u}}$ in the open neighborhood of $(\mathbf{x},\mathbf{g},\mathbf{y})$ where $\Delta=\det(A)$ does not vanish:
  \par\nointerlineskip
  \vbox{\baselineskip=0pt\small \begin{equation}
    \label{eq:incid-schur}
    \begin{cases}
      M/A = 0\\
      Y_{(2)} = (U_{(2)}-U_{(1)}A^{-1}B)^{-1} \\
      Y_{(1)} = \frac{-1}{\Delta} A^{-1}B Y_{(2)}
    \end{cases}
  \end{equation}}\nointerlineskip\noindent
and the truncated Jacobian matrix in $(\mathbf{X},\mathbf{Y})$ of this system can be written
\par\nointerlineskip
\vbox{\baselineskip=0pt\small  \begin{equation}
    \begin{bmatrix}
      \Jac_{\mathbf{X}}(M/A) & 0 & 0 \\
      \star & \Id{r_{1}(k-r_{1})} & \star \\
      \star & 0 & \Id{r_{1}(k-r_{1})}
    \end{bmatrix}
  \end{equation}}\nointerlineskip\noindent
  where $\Jac_{\mathbf{X}}(M/A)$ is the truncated Jacobian matrix in $\mathbf{X}$ of the $(k-r_{1})^{2}$ entries of $M/A$, which define $V_{r_{1}} \setminus \{(\mathbf{x,g}\mid \Delta=0)\}$ in $\CC^{n+t}$.
  By hypothesis~\hyp{6}, the ideal defined by the entries of $M/A$, which is a subideal of the ideal of all $(r_{1}+1)$-minors of $M$, is radical.
Since the Schur complement appears by multiplication with invertible matrices with entries in the localized ring $\QQ[\mathbf{X},\mathbf{g}]_{\Delta}$ (using the same notations as in the proof of Lemma~\ref{lemma:birat-incid}):
\par\nointerlineskip
\vbox{\baselineskip=0pt\small\begin{equation}
    \begin{bmatrix}
      \Id{r_{1}} & 0 \\
      -C & \Id{k-r_{1}}
    \end{bmatrix} \cdot
    \begin{bmatrix}
      A^{-1} & 0 \\
      0 & \Id{k-r_{1}}
    \end{bmatrix} \cdot
    M =
    \begin{bmatrix}
      \Id{r_{1}} & A^{-1}B \\
      0 & M/A
    \end{bmatrix},
  \end{equation}}\nointerlineskip\noindent
  Equations~\eqref{eq:incid-schur} describe the localization of $\langle\mathcal{F}_{r_{1},\mathbf{u}}\rangle$ in\linebreak $\QQ[\mathbf{X},\mathbf{g}]_{\Delta}$,
  so this ideal is radical as well. 
So we can use the Jacobian criterion on $V_{r_{1}}$ near $(\mathbf{x},\mathbf{g})$ and on $\mathcal{V}'_{r_{1},\mathbf{u}}$ near $(\mathbf{x},\mathbf{g},\mathbf{y})$.
Both Jacobians matrices have the same rank and both varieties have the same local codimension $(k-r_{1})^{2}$ (by Lemma~\ref{lemma:birat-incid} and hypothesis~\hyp{7}), so
  \begin{equation}
    K(\pi,V_{=r_{1}}) \cap\varphi(\mathcal{V}'_{r_{1},\mathbf{u}}) = K(\varphi,\mathcal{V}'_{r_{1},\mathbf{u}}) \cap \pi(V_{=r_{1}})
  \end{equation}
  
  The image $\varphi(\mathcal{V}'_{r_{1},\mathbf{u}}) \cap \pi(V_{=r_{1}})$ is a Zariski-open subset $\mathcal{O}_{\mathbf{u}}$ of  $\pi(V_{=r_{1}})$.
  It remains to prove that if $\mathbf{u}$ is sufficiently generic, then all irreducible components of\linebreak $K(\pi,V_{=r_{1}})$ meet this open subset.

  Let $\mathcal{C}_{1},\dots,\mathcal{C}_{a}$ be these irreducible components, and let $(\mathbf{x}_{1},\mathbf{g}_{1}) \in \pi^{-1}(\mathcal{C}_{1}), \dots, (\mathbf{x}_{a},\mathbf{g}_{a}) \in \pi^{-1}(\mathcal{C}_{a})$.
  For any $(\mathbf{x},\mathbf{g}) \in V_{=r_{1}}$, the proof of \cite[Prop.~4, Sec.~6]{HNSjoc16} shows  that there exists a non-empty Zariski-open subset $\mathcal{U}_{(\mathbf{x},\mathbf{g})} \subset \CC^{k(k-r_{1})}$ such that if $\mathbf{u} \in \mathcal{U}_{(\mathbf{x},\mathbf{g})}\cap \QQ^{k(k-r_{1})}$, then $(\mathbf{x},\mathbf{g}) \in \mathcal{O}_{\mathbf{u}}$; namely, $\mathcal{U}$ is the set of $\mathbf{u}$ such that
$    \rk
    \begin{bsmallmatrix}
      M(\mathbf{x},\mathbf{g})\\
      (u_{i,j})      
    \end{bsmallmatrix}
    = k$.
    Taking the finite intersection of the non-empty Zariski-open
    subsets $\mathcal{U}_{(\mathbf{x}_{i},\mathbf{g}_{i})}$ for
    $i \in \{1,\dots,a\}$ yields the wanted subset $\mathcal{U}$.
  \end{proof}
  
\subsection{Locus of rank exactly $r_{0}$}
\label{sec:Locus-corank-1}

Recall that by~\hyp{5}, $\pi(V \cap V_{=r_{0}})$ has codimension at least $1$, and that we want to compute a polynomial whose zeroes cover $K(\pi,V)$ and $\pi(V \cap B_{0})$. 
So we may multiply the result by the equation of one hypersurface covering $\pi(V \cap V_{=r_{0}})$, it will naturally cover $\pi(V \cap V_{=r_{0}}) \cap K(\pi,V)$ and $\pi(V \cap V_{=r_{0}}) \cap \pi(V \cap B_{0})$.

\noindent \textbf{Algorithm \algocorkone{}:}
\begin{description}
  \item[Input] $M \subset \QQ[\mathbf{X},\mathbf{G}]^{k \times k}$, $r_{0} \in \{1, \dots, k-1\}$ 
  \item[Output] $P_{1} \in \QQ[\mathbf{G}] \setminus \{0\}$ \textit{s.t.} $\pi(V \cap V_{=r_{0}}) \subset V(P_{1})$
  \item[Procedure] ${}$
  \begin{myalgo}
    \item $\algores \leftarrow 1$
    \item $F_{V,0} \leftarrow \{\text{$(r_{0}+1)$-minors of $M$}\}$
    \item $J \leftarrow \Jac_{\mathbf{X}}(F_{V,0})$
    \item $F_{V,1} \leftarrow F_{V,0} \cup \{\text{$(k-r_{0})^{2}$-minors of $J$}\}$
    \item Pick at random $u_{1},\dots,u_{k(k-r_{0})}=\mathbf{u} \in \QQ^{k(k-r_{0})}$
    \item $F_{0} \leftarrow \incidsys{k-r_{0}}{\mathbf{u}}$
    \item $\{\mathcal{M}_{1},\dots,\mathcal{M}_{\nummins}\} \leftarrow \{\text{$r_{0}$-minors of $M$}\}$
    \item For $i$ in $\{1,\dots,\nummins\}$ do
    \item \ind $F_{1} \leftarrow F_{0} \cup \mathcal{F}_{V,1} \cup \{\mathcal{M}_{1},\dots,\mathcal{M}_{i-1},u\cdot\mathcal{M}_{i}-1\}$
    \item \ind $G \leftarrow \funElim(F_{1},\{u,\mathbf{X},\mathbf{Y}\})$
    \item \ind Multiply $\algores$ by $1$ polynomial from $G$
    \item End for
    \item Return $\algores$
  \end{myalgo}
\end{description}
\subsection{Singularities} 
\label{sec:Incidence-varieties-1}

\noindent \textbf{Algorithm \algocrit{}:}
\begin{description}
  \item[Input] $M \subset \QQ[\mathbf{X},\mathbf{G}]^{k \times k}$, $r_{0} \in \{1,\dots,k-1\}$
  \item[Output] $P_{c} \in \QQ[\mathbf{G}] \setminus \{0\}$ \textit{s.t.} $K(\pi, V) \subset V(P_{c})$
  \item[Procedure] ${}$
  \begin{myalgo}
    \item $\algores \leftarrow \algocorkone(M,r_{0})$
    \item Pick at random $u_{1},\dots,u_{k(k-r_{0}+1)}=\mathbf{u} \in \QQ^{k(k-r_{0}+1)}$
    \item $F_{0} \leftarrow \incidsys{r_{0}-1}{\mathbf{u}}$
    \item $J \leftarrow \Jac_{\mathbf{X},\mathbf{Y}}(F_{0})$
    \item $F_{1} \leftarrow F_{0} \cup \{\text{$k(k-r_{0}+1)+(k-r_{0}+1)^{2}$-minors of $J$}\}$
    \item $G \leftarrow \funElim(F_{1},\{\mathbf{X},\mathbf{Y}\})$
    \item Multiply $\algores$ by $1$ polynomial from $G$
    \item Return $\algores$
  \end{myalgo}
\end{description}
\begin{proposition}
  \label{prop:correct-crit}
  Algorithm \algocrit{} is correct.
\end{proposition}
\begin{proof}
  By definition, $V \subset V_{r_{0}}$.  Using the decomposition
  $V_{r_{0}} = V_{=r_{0}} \cup V_{r_{0}-1}$, we decompose the variety
  $V$ as
  $V = \left(V \cap V_{=r_{0}}\right) \cup \left( V \cap V_{r_{0}-1}
  \right)$.
  
  The subspace $\pi(V \cap V_{=r_{0}})$ is covered by the output of \algocorkone{}, so we may restrict to $V \cap V_{r_{0}-1}$, which is the whole variety $V_{r_{0}-1}$ by Lemma~\ref{lemma:equiv-cork2-sing}.
  Recall that by hypothesis~\hyp{7}, $V_{r_{0}-1}$ is $t$-equidimensional.

  By Prop.~\ref{prop:equiv-incid-sing}, in order to compute $K(\pi,V_{r_{0}-1})$, we can compute polynomials whose zeroes cover $K(\varphi,\mathcal{V}'_{r_{0}-1,\mathbf{u}})$ with $\mathbf{u}$ sufficiently generic instead.

  By hypotheses~\hyp{6}, \hyp{7} and the proof of Prop.~\ref{prop:equiv-incid-sing}, $\mathcal{V'}_{r_{0}-1,\mathbf{u}}$ is $t$-equidimensional and $\incidsys{r_{0}-1}{\mathbf{u}}$ is a set of generators of its ideal, so we can use the Jacobian criterion to compute equations defining $K(\varphi,\mathcal{V}'_{r_{0}-1,\mathbf{u}})$ 
  .
\end{proof}
\subsection{Boundary}
\label{sec:Bound-determ-vari}

\noindent \textbf{Algorithm \algoboundary{}:}
\begin{description}
  \item[Input] ${}$
  \begin{itemize}
    \item $M \subset \QQ[\mathbf{X},\mathbf{G}]^{k \times k}$
    \item $r_{0}\in\{1,\dots,k-1\}$
    \item $H \subset \QQ[\mathbf{X},\mathbf{G}]$ (set of constraints on the variables)
  \end{itemize}
  \item[Output] $P_{b} \in \QQ[\mathbf{G}] \setminus \{0\}$ \textit{s.t.} $\pi(V \cap B_{0}) \subset V(P_{b})$
  \item[Procedure] ${}$
  \begin{myalgo}
    \item $\algores \leftarrow \algocorkone(M,r_{0})$
    \item Pick at random $u_{1},\dots,u_{k(k-r_{0}+1)}=\mathbf{u} \in \QQ^{k(k-r_{0}+1)}$
    \item $F_{0} \leftarrow \incidsys{r_{0}-1}{\mathbf{u}}$
    \item For $h$ in $H$ do
    \item \ind $F_{1} \leftarrow F_{0} \cup \{h\}$
    \item \ind $G \leftarrow \funElim(F,\{\mathbf{X},\mathbf{Y}\})$
    \item \ind Multiply $\algores$ by $1$ polynomial from $G$
    \item End for
    \item Return $\algores$
  \end{myalgo}
\end{description}

\begin{proposition}
  \label{prop:correct-boundary}
  Algorithm~\algoboundary{} is correct.
\end{proposition}
\begin{proof}
  As in Section~\ref{sec:Incidence-varieties-1}, we write:
  $V = \left(V \cap V_{=r_{0}}\right) \cup \left( V \cap V_{r_{0}-1}
  \right)$.
  Since $B_{0} = \bigcup_{h \in H} V(h)$, the intersection $V \cap \boundary{B}$ is contained in the union of the varieties $V(\langle  F \rangle + \langle h \rangle)$ for $h$ ranging over $H$, and the equation of the projections can be obtained with polynomial elimination.
\end{proof}
\begin{remark}
  For the real root classification problem, the subdivision is given by the product of the outputs of \algocrit{} and \algoboundary{}.
  In order to avoid repeating computations, we may skip the call to \algocorkone{} in either subroutine (but not both), and initialize $\algores$ to $1$ instead.
\end{remark}


\section{The contrast problem}
\label{section:WaterSingularities}
\subsection{The case of water}
\label{sec:case-water}


With the notations of Section~\ref{sec:intromodel}, the variety $V$ is the complex algebraic variety defined by 
\begin{equation}\small
 D = \pfrac{D}{y_1} = \pfrac{D}{y_2} = \pfrac{D}{z_1} = \pfrac{D}{z_2} = 0.
\end{equation}
With the notations of Section~\ref{sec:Algorithm}, we want to classify the singularities of the set of points where $M$ has rank at most $r_0=3$.
Our semi-algebraic constraints are that the solutions are within the Bloch ball, that is
  \begin{equation}\small
  \label{eq:4}
  {\mathcal B} :  \begin{cases}
    h_{1} = y_{1}^{2} + (z_{1}+1)^{2} \leq 1 \\
    h_{2} = y_{2}^{2} + (z_{2}+1)^{2} \leq 1.
  \end{cases}
\end{equation}
Since the equations are homogeneous in $\Gamma_{1},\Gamma_{2},\gamma_{1},\gamma_{2}$, and the parameters are supposed to be non-zero, we may normalize by setting $\gamma_{1}=1$.
In the case where the first matter is water, we further simplify by setting $\Gamma_1=\gamma_1=1$, leaving free the two parameters $\Gamma_2,\gamma_2$ corresponding to the second matter.
We recall that we also assume that $2\,\Gamma_2 \geq \gamma_{2}$ and
that $(\gamma_2,\Gamma_2) \neq (1,1) = (\gamma_1,\Gamma_1)$ (that is,
the second matter is not water).
\begin{theorem}
  \label{thm:GBsingul}
  Consider the 9 polynomials:\par\nointerlineskip\vskip-2mm
  \vbox{\baselineskip=0pt\small\begin{align}
    f_{1}&=\Gamma_{2}-1, \     f_{2}=3\,\Gamma_{2}-2\,\gamma_{2}-1 ,
    \\
    f_{3}&=3\,\Gamma_{2}^{2}-5\,\Gamma_{2}\gamma_{2}+\gamma_{2}^{2}+2\,\Gamma_{2}-2\,\gamma_{2}+1,
    \\
    f_{4}&=2\,\Gamma_{2}^{2}-5\,\Gamma_{2}\gamma_{2}+2\,\gamma_{2}^{2}-2\, 
    \Gamma_{2}+3\,\gamma_{2} ,
    \\
    f_5&=2\gamma_{2}^{3}- \left( 3\Gamma_{{2}}+11 \right) \gamma_{2}^{2}
    + \left( 9\Gamma_{{2}}+6 -3\Gamma_{2}^{2}\right) \gamma_{2}
    \\
    &\hspace{.35cm}
    +2\,\Gamma_{{2}} \left( \Gamma_{{2}}+2 \right)  \left( \Gamma_{{2}}-1 \right) ,
    \\
    f_{6}&=\Gamma_{2}-2\,\gamma_{2}+1, \  f_{7}=2\,\Gamma_{2}-\gamma_{2}-1,
    \\
    f_{8}&=\gamma_{2}-2+\Gamma_{2} , \ 
    f_{9}=2\,\Gamma_{2}^{2}-5\,\Gamma_{2}\gamma_{2}+2\,\gamma_{2}^{2}+1 .
  \end{align}
}\nointerlineskip\noindent  
The zeroes of their product divide the subset of $\RR^{2}$ defined by 
$2\,\Gamma_{2} > \gamma_{2} >0$
 into connected components where the cardinality of $V_{\RR} \cap \pi^{-1}(\gamma_{2},\Gamma_{2})$ is constant.
\end{theorem}

Let $\psi:(y_1,z_1,y_2,z_2)\mapsto (-y_1,z_1,-y_2,z_2)$ be the symmetry fixing $\Pi=\{y_1=y_2=0\}$, and
let us consider the semi-algebraic sets (see Fig. \ref{fig:bifurRealSingSimple}):\par\nointerlineskip\vskip-3pt
\vbox{\baselineskip=0pt\small\begin{align}
  {\mathcal G}_1^{-}&=\{   0 < \gamma_2 < 2\, \Gamma_2, \Gamma_2 < 1 , f_2 >0
  \} ,
  \\
  {\mathcal G}_1^{+}&=\{    0 < \gamma_2 < 2\, \Gamma_2, \Gamma_2 >1, f_2 < 0, f_4 >0 \} ,
  \\
  {\mathcal G}_2^{-}&=\{  0 < \Gamma_2 < 1 , f_6 >0, f_3 < 0\} ,
  \\
  {\mathcal G}_2^{+}&=\{  \Gamma_2 >1, f_6 < 0, f_5 >0 \} ,
  \\
  {\mathcal G} &= {\mathcal G}_1^{-} \cup   {\mathcal G}_1^{+} \cup   {\mathcal G}_2^{-} \cup {\mathcal G}_2^{+} .
\end{align}}
\nointerlineskip%

\begin{figure}
  \centering
  \centerline{\includegraphics{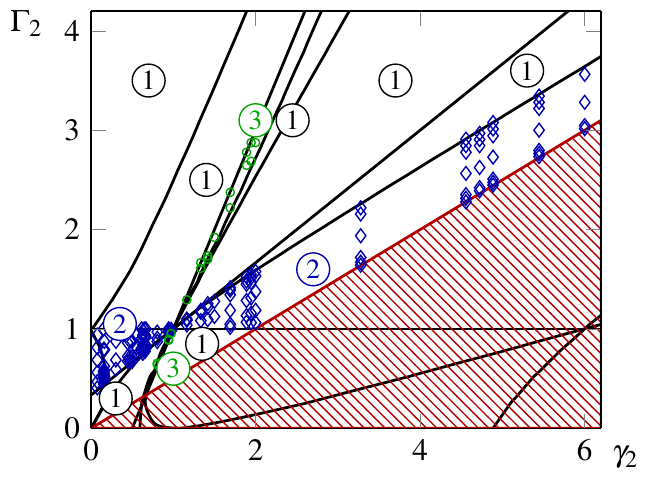}}
 \vskip-3mm
 \caption{Curves involved in the definition of the semi-algebraic set ${\mathcal G}$.
   The blue (\textit{resp.} green) sample points correspond to
   points in ${\mathcal G}_1^{-}\cup   {\mathcal G}_1^{+} $ (\textit{resp.}  ${\mathcal G}_2^{-}\cup   {\mathcal G}_2^{+} $).
   The circled numbers in each area correspond to the number of singularities in $B$ for parameters in the area.
   Parameters in the red area are physically irrelevant.}
  \label{fig:bifurRealSingSimple}
\end{figure}
\begin{theorem}
  \label{thm:DCAsingul}
  For all $(\gamma_2,\Gamma_2)$ such that $ 2\, \Gamma_2 > \gamma_2 > 0$, the center $O$ of the Bloch ball $\mathcal B$ is a singularity of $\{D=0\}$.
  And provided $(\gamma_2,\Gamma_2)\in {\mathcal G}$, there exist at most two other singularities:
  \begin{enumerate}
    \item provided $(\gamma_2,\Gamma_2)\in {\mathcal G}_1^{-}\cup   {\mathcal G}_1^{+} $ there is one other singularity lying on $\Pi \cap \mathcal B$;
    \item provided $(\gamma_2,\Gamma_2)\in {\mathcal G}_2^{-} \cup {\mathcal G}_2^{+} $ there are two other singularities in $\mathcal B$, $\psi$-symmetric, outside $\Pi$.
  \end{enumerate}
\end{theorem}
The configuration is illustrated in Figs.~\ref{fig:bifurRealSingSimple} and~\ref{fig:bifurRealSingW}.
Observe that the number of singularities inside $\mathcal B$ is an invariant of the contrast problem. 
Two of the pairs of biological matters studied in~\cite{MCRF13}, water-cerebrospinal fluid (normalized parameters $[\gamma_2=\frac54,\Gamma_2=\frac{25}{3}]$)
and water-fat (normalized parameters $[\gamma_2=\frac{25}{2},\Gamma_2=25]$) correspond to points outside ${\mathcal G}$, and their invariant is $1$ in both cases
(see Fig.~\ref{fig:bifurRealSingW}). 
But our results give answers to our guiding questions: there exist pairs of matters for which this algebraic invariant can differ; and any pair $(\text{water},\text{matter})$ belongs to one of $3$ classes, depending on whether the number of singularities inside ${\mathcal B}$ is $1$, $2$ or $3$.

\begin{figure}
  \centering
  \centerline{\includegraphics{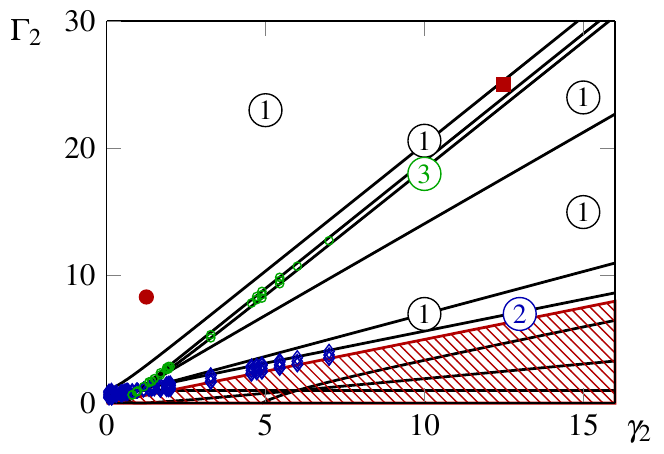}}
\vskip-3mm
\caption{Positions of the parameters corresponding to the pairs water-cerebrospinal fluid (red circle) and water-fat (red square) and the set ${\mathcal{G}}$ (with the same conventions as in Fig.~\ref{fig:bifurRealSingSimple}).
For both these pairs, there is only $1$ singularity in $\mathcal{B}$.}
  \label{fig:bifurRealSingW}
\end{figure}

{\small\begin{table*}
  \centering\small
  \begin{tabular}{lrrrrrr}
    \toprule
    System & \pbox{\textsf{RegularChain} \\ (direct)} & \pbox{Gröbner \\ (direct)} & \pbox{\textsf{RegularChain} \\ (new algo.)} & \pbox{\textsf{FGb}\\ (new algo.)} & \textsf{F5} (new algo.) & CAD \\
    \midrule
    Water & \SI{1600}{\second} & \SI{100}{\second} & & \SI{10}{\second} & \SI{1}{\second}& \SI{50}{\second}\\
    General & $>$\SI{24}{\hour} & $>$\SI{24}{\hour} & $90\times \SI{200}{\second}$ & $46 \times \SI{200}{\second}$ & \SI{110}{\second} & \SI{4}{\hour} (projection step) \\
    \bottomrule
  \end{tabular}
  \caption{Timings}
  \label{tab:timings}
\end{table*}}

\label{sec:Proof-Theor-refthm:G}

\begin{proof}[of Theorem~\ref{thm:GBsingul}]
  Let $V_{=3} = \{p \in \CC^{4}\times\RR^{2} \mid \rk(M) = 3\}$ and $V_2 = \{p \in \CC^{4}\times\RR^{2} \mid \rk(M) < 3\}$, where
$p=(y_1,y_2,z_1,z_2,\gamma_{2},\Gamma_{2})$.
  We apply the strategies described in Section~\ref{sec:Algorithm}.
  
  We study the generic case $V_2 \cap V$ first.
  This set does cover a dense subset of $\RR^{2}$.
  Its intersection with the boundary of $\mathcal{B}$ is given by the vanishing of either $h_{1}$ or $h_{2}$.
  The projection on $(\Gamma_{2},\gamma_{2})$ of the set of points of $V_2 \cap V$ such that $h_{1}=0$ is described by
    $0 = \gamma_{2} f_{1}^{2} f_{2} f_{3}$
  which gives us polynomials $f_{1}$, $f_{2}$ and $f_{3}$.

  The projection on $(\Gamma_{2},\gamma_{2})$ of the set of points of $V_2 \cap V$ such that $h_{2}=0$ is described by
    $0 = \left( 2\,\Gamma_{2}-\gamma_{2} \right) f_{1}^{2} f_{4} f_{5}$
  which gives us new polynomials $f_{4}$ and $f_{5}$.
  
  Next, we consider the incidence variety $\incidvar{2}$ associated with the matrix $M$:\par\nointerlineskip
\begin{equation}\small
    M \cdot
    \begin{psmallmatrix}
      \lambda_{1,1} & \lambda_{1,2} \\
      \lambda_{2,1} & \lambda_{2,2} \\
      \lambda_{3,1} & \lambda_{3,2} \\
      \lambda_{4,1} & \lambda_{4,2}
    \end{psmallmatrix}
    =
    \begin{psmallmatrix}
      0 & 0 \\
      0 & 0 \\
      0 & 0 \\
      0 & 0       
    \end{psmallmatrix}
  \end{equation}
with random linear equations ensuring that the matrix $(\lambda_{i,j})$ has rank $2$.
  
  Out of the surface $\gamma_{2} = 0$, this affine variety is a complete intersection (it has dimension $2$ and it is given by $9$ equations in $11$ variables, including the saturation by $\gamma_{2}$).
  The set of critical values of $\pi$ is described by
    $0 = \left( 2\,\Gamma_{2}-\gamma_{2} \right)(\Gamma_{2}+1) f_{1}^{2} f_{6}^{2} f_{7}^{2}$
  which gives us new polynomials $f_{6}$ and $f_{7}$ ($\Gamma_{2}+1$ has no solutions within our constraint range).
  
  This completes the study of $V \cap V_2$.
  We now move on to the study of $V \cap V_{=3}$.
  As described in the algorithm, we define the incidence variety of rank $3$ of $M$, and we saturate successively by the $3$-minors of $M$.
  Only the first of these subcases is nonempty, and it is described by
    $0 = \left( 2\,\Gamma_{2} - \gamma_{2} \right) f_{8} f_{9} $
  which gives us $f_{8}$ and $f_{9}$.
\end{proof}

\begin{proof}[of Theorem~\ref{thm:DCAsingul}]
  Observe first by means of a trivial evaluation that $O$ is a singularity of $\{ D =0 \}$. We now focus on singularities in ${\mathcal B}^{*}={\mathcal B}\setminus \{O\}$.
  Theorem~\ref{thm:GBsingul} provides a list of 9 polynomials to which we add our constraints 
  $2\,\Gamma_{2} \ge \gamma_{2} > 0$.
  Let $\xi=\gamma_2\,\Gamma_2\,\left( \gamma_2 - 2\, \Gamma_2 \right) \, \prod_{i=1}^{9} f_i$.
The complementary of $\{ \xi=0 \}$ is the union of a sequence of connected open semi-algebraic sets where the number of singularities is constant.
The routine \textsf{CylindricalAlgebraicDecompose} of the Maple package \textsf{RegularChains[Semi\-AlgebraicSetTools]}
 provides 1533~sample points
 .
Excluding those at which $\xi$ vanishes and those outside our physical constraints domain, remains a set $K_c$ of $570$ points.
At each point of $K_c$ we locate the singularities by computing a Gr\"obner basis.

  We get $187$ points  of $K_c$
  such that there exists at least one singularity in ${\mathcal B}^{*}$.
  We have a set $K_s$ of $31$ points, each of them corresponding to a couple of $\psi$-symmetric singularities outside the symmetry plane $\Pi$, 
  and a set $K_p$ of $156$ points corresponding to a unique singularity on  $\Pi \cap {\mathcal B}^{*}$. 
  For parameters at which $\xi$ does not vanish, the number of singularities in ${\mathcal B}^{*}$ is at most two.

  Points of $K_s$ (\textit{resp.} $K_p$) are represented in green  (\textit{resp.} blue)  in Figs. \ref{fig:bifurRealSingSimple}, \ref{fig:bifurRealSing} and \ref{fig:bifurRealSingLoupe}.
  Let us evaluate on $K_c$ the condition ($\Gamma_2 < 1$,  $f_2> 0$, $f_4 <0$) or ($\Gamma_2 >1$, $f_2 < 0$, $f_4 >0$). 
  Indeed the set of points of $K_c$ satisfying this condition coincides with $K_p$.
  This proves item 1).
  The proof of item 2) is similar.
\end{proof}

\begin{figure}
  \centering
  \centerline{\includegraphics{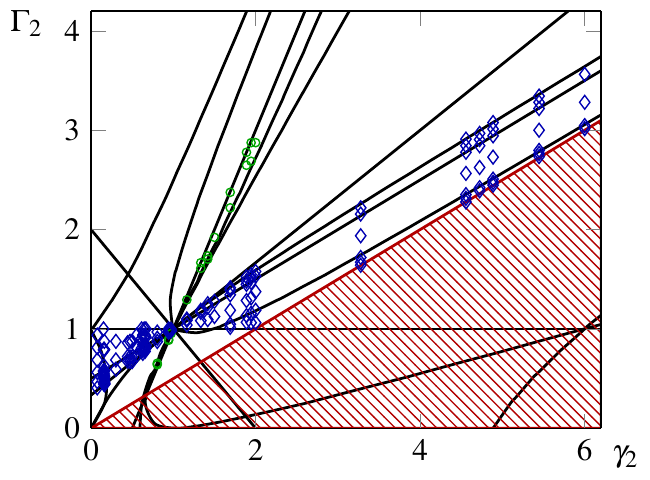}}
\vskip-3mm
\caption{The curves involved in the decomposition of 
  the parameter space (with the same conventions as in Fig.~\ref{fig:bifurRealSingSimple}).}
  \label{fig:bifurRealSing}
\end{figure}
\begin{figure}
  \centering
  \centerline{\includegraphics{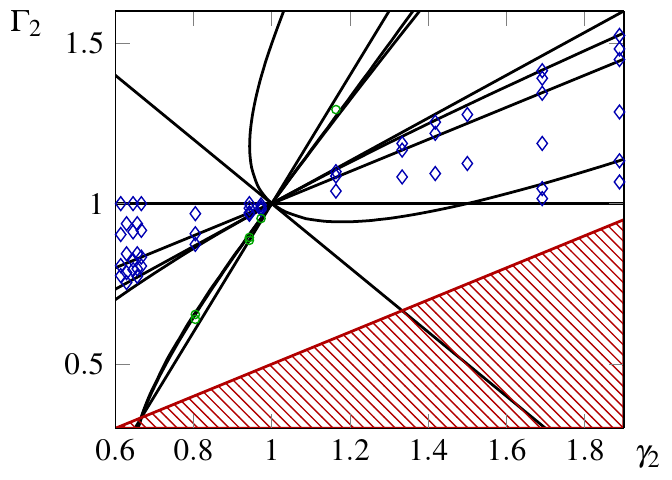}}
\vskip-3mm
\caption{Magnification of Fig.~\ref{fig:bifurRealSing} near $(1,1)$.}
  \label{fig:bifurRealSingLoupe}
\end{figure}

\subsection{The general case}
\label{sec:General-case}

The variety $V$ and the semi-algebraic set ${\mathcal B}$ are defined as in the previous section.
We normalize again by $\gamma_{1}=1$, we assume that 
$2\,\Gamma_{1} \ge 1$, $2\,\Gamma_2\ge \gamma_2>0 $, $(\gamma_2,\Gamma_2) \neq (1,\Gamma_{1})$, 
and that $\Gamma_{1} \neq 1$, $\Gamma_{2} \neq \gamma_{2}$ (case of water).


\begin{theorem}\label{thm:general-res}
  Splitting the subset of $\RR^{3}$ defined by  $2\,\Gamma_{2} > \gamma_{2}>0$ and $2\,\Gamma_{1} > 1$ into open subsets where the number of real singularities of $V$ in the fibers is constant, can be done by cutting out $12$ irreducible surfaces, consisting of $5$ planes,
 $3$ quadrics,
two surfaces of degree 9 and one of degree 14.
\end{theorem}

These polynomials were obtained by applying the algorithms from Section~\ref{sec:Algorithm} to our system.
The elimination steps were done using both Gröbner bases with \textsf{FGb} or with \textsf{F5}, and with triangular sets with \textsf{RegularChains}.
Table~\ref{tab:timings} presents some timings for these methods (for computations done with interpolation, we give the results as $a \times b$ where $a$ is the interpolation degree and $b$ the time taken for each specialized computation).
It shows that our algorithm, implemented with either Gr{\"o}bner bases or regular chains, is more efficient than the direct (general) approach, and  it allows to deal with the previously untractable general case.
Furthermore, the implementation using Gr\"obner bases appears to be faster.





{
  \bibliographystyle{plain}
  \bibliography{biblio}
}

\end{document}